\documentclass[letterpaper, 10 pt, conference]{ieeeconf}
\IEEEoverridecommandlockouts    
\usepackage{graphicx}
\usepackage{subfig}

\usepackage[T1]{fontenc}
\usepackage[latin9]{inputenc}

\newtheorem{lemma}{Lemma}
\newtheorem{proposition}{Proposition}
\newtheorem{corollary}{Corollary}
\newtheorem{theorem}{Theorem}
\newtheorem{assumption}{Assumption}

\newtheorem{claim}{Claim}
\usepackage{cite}
\usepackage{amsmath}
\usepackage{amstext}
\usepackage{amssymb}
\usepackage{tikz}
\usepackage[mathscr]{euscript}

\usepackage{amsfonts}
\usepackage{enumerate}
\usepackage[final]{pdfpages}
\usepackage{csquotes}
\usepackage{tikz}
\usetikzlibrary{arrows.meta,shapes.geometric}

\let\labelindent\relax
\DeclareOldFontCommand{\rm}{\normalfont\rmfamily}{\mathrm}

\usepackage{graphicx}      
\usepackage{enumitem}
\usepackage{mathdots}

\usepackage{mathtools,xparse}

\usepackage{hyperref}
 \hypersetup{
    colorlinks=true,
    linkcolor=magenta,
    filecolor=magenta,
    urlcolor=cyan,
    citecolor=blue,
}
\DeclareMathOperator{\diag}{\texttt{diag}}
\definecolor{ForestGreen}{RGB}{34,139,34}

\DeclareMathOperator{\sgn}{sign}

\DeclareMathOperator{\dete}{determinant}
\def\qed{\hfill $\Box$}

\date{\today}

\begin{document}
\author{Sebin Gracy, Mengbin Ye, and Brian D.O. Anderson
\thanks{Sebin Gracy  is with the Department of Electrical  Engineering and Computer Science, South Dakota School of Mines, Rapid City, SD,USA (\texttt{sebin.gracy@sdsmt.edu}). Mengbin Ye is with the School of Mathematical Sciences, Adelaide University, Australia \texttt{ben.ye@adelaide.edu}. Brian D.O.~Anderson is with the School of Engineering, Australian National University, Canberra, Australia, \texttt{(brian.anderson@anu.edu.au)}. 
}.}
\title{\LARGE \bf On a Class of Decentralized Feedback Controllers for the Networked Bivirus SIS Model}
\maketitle
\begin{abstract}
 Recent work has established many properties of systems modeling the spread of two competing viruses in a population, and for networks with multiple connected populations with different spreading properties (e.g. men and women). This work considers the introduction of a class of nonlinear decentralized feedback controls aimed at reducing the fractions of populations infected at an endemic equilibrium. One surprising conclusion is that in some circumstances, new types of equilibria can arise. They can be viewed as an outcome of a transcritical bifurcation which is never observed for uncontrolled systems.  In addition, we show that the controlled system is strongly monotone, and  an unstable boundary equilibrium of the uncontrolled system cannot be stabilized using the class of decentralized state feedback controllers considered in this paper.  
\end{abstract}

\section{Introduction}

A number of papers have considered the dynamics of systems modeling a population in which two competing viruses ---where an individual can be infected by at most one virus at a particular point in time ---are circulating \cite{castillo1989epidemiological,prakash2012winner,bremermann1989competitive,gjini2016direct}. 
Extensions have considered multiple populations which have different infection or healing parameters, connected over a network. Multiple populations might arise from considerations such as age, gender, ethnicity, etc as well as geography \cite{sahneh2014competitive,liu2019analysis,wang2012dynamics,ye2022bivirus_survey,janson2024competitive,anderson2023equilibria,gracy2022siads,gracy2026,santos2015bi,pare2021multi,yang2017bi,li2003coexistence}. 
{As illustration, we recall the well-known susceptible-infected-susceptible (SIS) model for the single population setting (multipopulation equations appear later).
In these equations, for $k=1,2$ the variables $x^k$ denote the \textit{fraction} of the population infected with virus~$k$, and the constants 
$\delta^k,b^k$
denote positive \textit{healing parameters} and positive \textit{infection parameters} respectively.

\begin{align}\label{eq:uncontrolled_single}
\dot x^1(t)&=[-\delta^1+(1-x^1(t)-x^2(t))b^1]x^1(t)\\\nonumber
\dot x^2(t)&=[-\delta^2+(1-x^1(t)-x^2(t))b^2]x^2(t)   
\end{align}
Given an initial outbreak, $x^k(0)\neq 0, k=1,2$, the situation of real interest arises when infections due to either virus persist (i.e., are endemic). This is the case just when for both viruses the infection/recovery ratios both exceed 1:
\begin{equation}\label{eq:ratio_cond}
b^1/\delta^1>1;\quad b^2/\delta^2>1.   
\end{equation}
In the contrary case, one or both viruses simply die out. For generic parameters and with \eqref{eq:ratio_cond}, there are three equilibria, viz. the 
disease-free equilibrium (DFE)
of $\bar x^1 = \bar x^2 = 0$, and two endemic equilibria of $\bar x^1=1-\delta^1/b^1, \bar x^2=0$ and $\bar x^1=0, \bar x^2=1-\delta^2/b^2$ \cite{zino2024modeling}. Only one of the latter two is (almost) globally stable. In case $b^1/\delta^1>b^2/\delta^2>1$, from any initial condition with $x^k(0)\neq 0, k=1,2$ one of two endemic (steady state) equilibria is reached, viz.  $\bar x^1=1-\delta^1/b^1, \bar x^2=0$. Conversely, if $b^2/\delta^2>b^1/\delta^1>1$, the other endemic equilibrium is reached, being $\bar x^1=0, \bar x^2=1-\delta^2/b^2$. The nonattractive endemic equilibrium is a saddle and will be the limit of trajectories which start with a zero initial condition for one virus.  For the nongeneric case $b^1/\delta^1=b^2/\delta^2>1$, 
a line of equilibria joining the two points $(1-\delta^1/b^1,0)$ and $(0,1-\delta^1/b^1)$ exists. To summarize then, when one virus is stronger i.e., $\delta^i/b^i > \delta^j/b^j$ for $i,j \in \{1,2\}$ and $i\neq j$, the stronger virus~$i$ determines the endemic equilibrium, driving out the other virus; the outcome is described as a `winner take all' phenomenon. 

Reduction of the infected fraction associated with an endemic equilibrium (especially the almost globally attractive one) is an obviously reasonable public policy objective, and nonpharmaceutical or pharmaceutical controls may be introduced. We shall assume they take the form of raising the healing parameter, with \textit{the extent of the raising increasing monotonically with  the infected fraction}. This is not unreasonable. 
For instance, during  COVID-19, a surge in cases led to an increased risk perception in the unvaccinated population \cite{lee2024dynamics}, which in turn resulted in a greater willingness towards getting vaccinated \cite{akel2022study}. Studies have shown that adoption of vaccines such as mRNA resulted in $6.4$ fewer total days of
symptoms, thus leading to a faster recovery.
More precisely, we study the modified system
\begin{small}  
\begin{align}\label{eq:controlledsingle}
\dot x^1(t)&=[-\delta^1+(1-x^1(t)-x^2(t))b^1-h^1(x^1(t))]x^1(t)\\\nonumber
\dot x^2(t)&=[-\delta^2+(1-x^1(t)-x^2(t))b^2-h^2(x^2(t))]x^2(t)  
\end{align}
\end{small}
with each $h^k(\cdot)$ nonnegative and  monotone increasing in its argument. Note that the uncontrolled and controlled equation differ by a term which is overall nonlinear in the state due to this monotonicity. 

The controller introduced in~\eqref{eq:controlledsingle} is inspired by the one studied in \cite{ye2021_PH_TAC}, which considered   the single-virus case. It was found that a particular class of decentralized feedback controllers could not eliminate the endemic equilibrium entirely, but would reduce the equilibrium infection level for each population. It is thus natural to seek to extend this to two viruses, and determine to what extent the same improvement in endemic control can be achieved. However, the behavior of the uncontrolled SIS model is more complex with two viruses as compared to one, due to the coupled spreading and competitive dynamics, and the extension turns out to be far from straightforward.


The key issue of concern in this paper is to understand what adjustment arises to the equilibria (and associated convergence properties). 
A surprising outcome is that, with two competing viruses, the `winner take all' equilibrium picture may fail. Indeed, not only is a coexistence equilibrium now generically possible $(\bar x^1\neq 0, \bar x^2\neq 0)$, it can be (almost) globally attractive. In context, this suggests that while a certain class of decentralized feedback control always improves the endemic situation for the single virus problem as noted above, it can negatively impact a bivirus system by allowing both viruses to persist as opposed to one always becoming extinct.

When 
$h^i(.)$ (for $i=1,2$) are taken to be linear for the sake of simplicity, 
the transition from a `winner take all' equilibrium pattern to the coexistence possibility is associated with a transcritical bifurcation associated with change of the scaling parameter. Very simply put, a transcritical bifurcation is said to occur when, as a result of a smooth parameter change,  an equilibrium whose position varies smoothly with the parameter switches its stability properties (e.g., a stable (resp. unstable) equilibrium point becomes unstable (resp. stable)). See \cite{strogatz2024nonlinear,perko2013differential} for an in-depth understanding. 
Our primary intention in this paper is to conclusively establish that the single population model exhibits 
the phenomenon of transcritical bifurcation (Section~\ref{sec:n=1:formal:analysis}). 
Our other major contribution is to show that, for the multiple population case, an unstable boundary equilibrium for the uncontrolled system cannot be stabilized using the class of decentralized state feedback controllers considered in this paper (Theorem~\ref{thm:merged}). 

\textit{Outline of paper:} 
The multipopulation model being investigated in this paper is formally introduced in Section~\ref{sec:model}. 
In Section~\ref{sec:n=1case}, we show  that for the single population case, as the gain of the controller is increased, the system exhibits a transcritical bifurcation,
whereas Section~\ref{sec:equi:analysis} deals with the multipopulation case and focuses on identifying conditions for stability of some of the different kinds of equilibria. Finally, we summarize the paper and highlight some directions of possible interest to the wider community in Section~\ref{sec:conclusion}

{\em Notation}: 
For any positive integer $n$, the set $\{1,2,\ldots,n\}$ is denoted by $[n]$, while 
 $\mathbf{0}$, $\mathbf{1}$  denote vectors with entries all equal $0$ and $1$, respectively.
We write $A\succ 0$  ($A\succeq 0$) when the matrix $A$ is symmetric positive definite (resp. positive semidefinite). Also, $A\prec 0$ (resp. $A\preceq 0$) indicates that matrix $A$ is negative definite (resp. negative semidefinite).
We use $\mathbb{R}$ to denote the set of real numbers
and the set of nonnegative real numbers is denoted by $\mathbb{R}_+$.  For a vector $x$ we denote the diagonal square matrix with $x$ along the diagonal by $\diag(x)$. For any two real vectors $a, b \in \mathbb{R}^n$ we write $a \geq b$ if $a_i \geq b_i$ for all $i \in [n]$, $a>b$ if $a \geq b$ and $a \neq b$, and $a \gg b$ if $a_i > b_i$ for all $i \in [n]$. 
For a square matrix $M$, the  spectral radius, largest eigenvalue abscissa and determinant are denoted by $\rho(M)$,  $s(M)$ and $\det(M)$, respectively.

A real square matrix $A$ is called Metzler if all its off-diagonal entries are nonnegative.
If $A(=[a_{ij}]_{n\times n})$ is a nonnegative matrix, then $\rho(A)$  decreases monotonically with a decrease in $a_{ij}$ for  any
 $i,j \in [n]$. The matrix $A$ is reducible if, and only if, there is a permutation matrix $P$ such that $P^\top AP$ is block upper triangular; otherwise $A$ is said to be irreducible. If a nonnegative $A$ is irreducible and $y=Ax$, then  $x > \textbf{0}$, implies $y > \textbf{0}$, and $y$ cannot have a zero  in every position where $x$ has a zero.

\section{Model}\label{sec:model}
We start by considering  the following (uncontrolled) dynamical system. It corresponds to a multipopulation generalization of \eqref{eq:uncontrolled_single}, with $x^k_j$ denoting the fraction of population $j\in[n]$ infected with virus $k$:
\begin{equation}\label{eq:bivirus}
\begin{array}{rcl}
   \dot{x}^1(t) &=  \Big{(} \big{(} I - (X^1+X^2) \big{)} B^1 - D^1 \Big{)} x^1(t), \\
      \dot{x}^2(t) &=  \Big{(} \big{(} I - (X^1+X^2) \big{)} B^2 - D^2 \Big{)} x^2(t), 
\end{array} 
\end{equation} 
where $x^1, x^2 \in \mathbb{R}^n$; $D^k, B^k$ for $k=1,2$ are of appropriate dimensions, and $X^k=\diag{(x^k)}$ for  $k=1,2$. For $k=1,2$, we define $D^k:=\diag{(\delta_1^k, \delta_2^k, \hdots, \delta_n^k)}$, with the $j$-th diagonal entry denoting the healing parameter for virus $k$ and population $j$. The matrices $B^k=[b^k_{ij}]_{n \times n}$, for $k=1,2$, capture the spreading of virus~$k$ between populations $i$ and $j$. 
\\ Following standard practice, we introduce:

\begin{assumption} \label{assum:base}The matrices $D^k$ are 
positive.
The matrices $B^k$ are nonnegative. 
\end{assumption}

Spreading can be represented by a two-layer graph $\mathcal G = \{\mathcal V, E_1, E_2\}$, where $\mathcal V= \{1,2, \hdots, n\}$ are the set of nodes (populations), and the edge sets $E_1$ and $E_2$ determine the contact spreading network for virus $1$ and virus $2$, respectively \cite{sahneh2014competitive}. 

In order for~\eqref{eq:bivirus} to better reflect realistic scenarios, and as is common in the literature, we also adopt:

\begin{assumption} \label{assum:irreducible}
The matrix $B^k$, for $k=1,2$ is irreducible, or equivalently  \cite[Theorem~4.3]{bullo2024lectures}
the two layers of the multi-layer network $\mathcal G$ are separately strongly connected. 
\end{assumption}


Inspired by the work in \cite{ye2021_PH_TAC}, we now consider the general class of decentralised,
local state feedback controllers of the form
\begin{equation}\label{eq:dist.contrller}
    u_i^k(t) = h_i^k(x_i^k(t)), 
\end{equation}
where $h_i: [0, 1] \rightarrow \mathbb{R}_{\geq 0}$ is bounded, smooth and monotonically nondecreasing, satisfying $h_i^k(0) = 0$ for all $ i \in [n]$ and $k \in [2]$. 
With this control introduced to increase the healing rate as infections increase, the new dynamics become 
\footnotesize
\begin{equation}\label{eq:bivirus:controlnew}
\begin{array}{rcl}
   \dot{x}^1(t) &=  \Big{(} \big{(} I - (X^1(t)+X^2(t)) \big{)} B^1 - D^1 -H^1(x^1(t)) \Big{)} x^1(t) \\
      \dot{x}^2(t) &=  \Big{(} \big{(} I - (X^1(t)+X^2(t)) \big{)} B^2 - D^2 -H^2(x^2(t)) \Big{)} x^2(t) 
\end{array} 
\end{equation} 
\normalsize
Define the regions
\begin{align}\label{eq:domain}
  \mathcal{D}:&= \{x^k(t) \in [0,1]^n,  \forall k \in [2]\mid 
  x^1+x^2 \leq \textbf{1}\} \nonumber \\  \mathcal{D}^k:&= \{x^k(t) \in [0,1]^n\}. 
\end{align}
Fractions of populations in the real world clearly lie in the interval $[0,1]$. Unsurprisingly then a standard result is that $\mathcal D$ is positively invariant with respect to \eqref{eq:bivirus}~\cite[Lemma~8]{liu2019analysis}.
The same holds true for the controlled system, as set out in the following lemma. The lemma conclusion actually is a little tighter, indicating circumstances where strict avoidance of the boundary occurs. The proof of the positive invariance property 
for the controlled system is almost identical with that for the uncontrolled system; the aforementioned  tighter conclusion can be obtained by a slight modification of the proof of \cite[Lemma~3.3]{ye2021convergence}, and hence, we omit the proof.
\begin{lemma}\label{lem:pos:inv:D:control}
Consider system~\eqref{eq:bivirus:controlnew} under Assumption~\ref{assum:base}. Suppose that for all $i \in [n]$ and $k \in [2]$, $h_i^k
: [0, 1] \rightarrow \mathbb{R}_{\geq 0}$ is
bounded, smooth and monotonically nondecreasing, satisfying $h_i^k(0) = 0$. 
Suppose that the initial state lies in $\mathcal D$, i.e.  $x_i^1(0), x_i^2(0), x_i^1(0)+x_i^2(0) \in [0,1]$; then $x_i^1(t), x_i^2(t), x_i^1(t)+x_i^2(t) \in [0,1]$ for all $t \in \mathbb{R}_{\geq 0}$ and $i \in [n]$. Moreover, 
{\color{black} if the initial state $x^1(0),x^2(0)$ lies in the interior of $\mathcal D$, then for all finite $t>0$, $\textbf{0} \ll x^k(t) \ll \mathbf{1}$ for $k=1,2$ and $x^1(t)+x^2(t) \ll \mathbf{1}$. 
}
\end{lemma}

The next two  lemmas establish important properties of each equilibrium of system~\eqref{eq:bivirus:controlnew}. Apart from coexistence equilibria (those in the interior of $\mathcal D$), equilibria on the boundary of $\mathcal D$ are restricted to being at the 
DFE, or as an endemic equilibrium associated with a single virus, i.e. the other virus is extinct.  Further, whenever one population at steady state is infected with a virus, all populations must be infected with the same virus. Again, this conclusion is the same as one applying the uncontrolled case. The second of the two lemmas, which is a restatement of Theorem 4 of \cite{ye2021_PH_TAC}, states that single virus endemic equilibria always exist and for each virus are unique.

\begin{lemma} \label{lem:equi_non-zero_nonone:1}
Consider system~\eqref{eq:bivirus:controlnew} under     
Assumptions~\ref{assum:base} and~\ref{assum:irreducible}. 
Suppose that for all $i \in [n]$ and $k \in [2]$, $h_i^k
: [0, 1] \rightarrow \mathbb{R}_{\geq 0}$ is bounded, smooth and monotonically nondecreasing, satisfying $h_i^k(0) = 0$. 
 If $\bar{x} = (\bar{x}^1,  \bar{x}^2) \in \mathcal{D}$ is an equilibrium of~\eqref{eq:bivirus:controlnew}, then, for each $k \in [2]$, either $\bar{x}^k = \textbf{0}$, or $\textbf{0} \ll \bar{x}^k \ll \textbf{1}$. Moreover, 
$\textstyle \sum_{k=1}^2 \bar{x}^k \ll \textbf{1}$.
\end{lemma} 
\textit{Proof:} The proof is identical to that of \cite[Lemma~6]{janson2024competitive}, in view of the hypothesis that for  all $i \in [n]$ and $k \in [2]$, $h_i^k
: [0, 1] \rightarrow \mathbb{R}_{\geq 0}$ is bounded, smooth and monotonically nondecreasing, and $h_i^k(0) = 0$. The details are omitted in the interest of space.~\qed

\begin{lemma}\label{lem:yeLotka}
    Adopt the same hypothesis as Lemma \ref{lem:equi_non-zero_nonone:1}. In addition, suppose that 
    $s(-D^k+B^k)>0$ for $k=1,2$.
    There exist unique boundary equilibria $(\bar x^1,{\bf{0}})$ and $({\bf{0}},\bar x^2)$. 
    If $x^j(0) = {\bf 0}$ for some $j$, then $x^i(t) \to \bar x^i$ for $j\neq i$.
\end{lemma}
Note that the lemma is \textit{not} claiming that $(\bar x^1, {\bf 0})$ and $({\bf 0},\bar x^2)$ are necessarily stable equilibria for the bivirus system in \eqref{eq:bivirus:controlnew}. It does imply they are either saddles or stable. 

In the uncontrolled case, the Jacobian of the associated system is key to establishing that the underlying system is monotonic: the monotonicity then implies trajectory ordering properties, and a key conclusion that from almost all initial conditions, convergence will occur to an attractive equilibrium.  That is, suppose that $(x_A^1(0), x_A^2(0))$ and $(x_B^1(0), x_B^2(0))$ are two initial conditions in $\textrm{int}(D)$ satisfying i) $x_A^1(0)>x_B^1(0)$ and ii) $x_A^2(0)<x_B^2(0)$. Since the  system is monotone, it follows that, for all $t \in \mathbb{R}_{\geq 0}$, i) $x_A^1(t)\gg x_B^1(t)$ and ii) $x_A^2(t)\ll x_B^2(t)$. 
To draw similar conclusions for the controlled system, observe first that the Jacobian of system~\eqref{eq:bivirus:controlnew} at an arbitrary point $(x^1, x^2) \in \mathcal D$ is as follows:
\begin{equation}\label{eq:jacobian}
    J(x^1, x^2)=
    \begin{bmatrix}
        J_{11}&J_{12} \\
        J_{21}&J_{22}
    \end{bmatrix}
\end{equation}
where
\begin{align}\label{eq:jacobianblocks}
   J_{11}&=-D^1-H^1(x^1)-(H^1(x^1))^\prime X^1 \nonumber \\ &~~~~+(I-X^1-X^2)B^1 -\diag(B^1x^1) \\\notag
   J_{12}&=-\diag(B^1x^1)\\\notag
   J_{21}&=-\diag(B^2x^2)\\\notag
   J_{22}&=-D^2-H^2(x^2)-(H^2(x^2))^\prime X^2 \nonumber \\ &~~~~+(I-X^1-X^2)B^2-\diag(B^2x^2) \notag
\end{align}

It is immediate that, as for the case treated in 
\cite{ye2021convergence}
when there is no control, (corresponding to $H^k(x^k)$ and its derivative being replaced by zero in the blocks of $J$), the sign pattern is special. In fact, with the definition $P=\diag(I_n,-I_n)$, the matrix $PJP$ is immediately seen to be a Metzler matrix. This preservation of the sign pattern when introducing the control means the following property, known from the control-free case, is valid.  The proof for system \eqref{eq:bivirus:controlnew} differs only very slightly from that in the control free case, and will be omitted.

\begin{theorem}\label{thm:monotonenew}
    Consider system~\eqref{eq:bivirus:controlnew} under Assumption~\ref{assum:base} and~\ref{assum:irreducible}.  Suppose that for all $i \in [n]$ and $k \in [2]$, $h_i^k
: [0, 1] \rightarrow \mathbb{R}_{\geq 0}$ is
bounded, smooth and monotonically nondecreasing, satisfying $h_i^k(0) = 0$. Then system~\eqref{eq:bivirus:controlnew} is strongly monotone.
\end{theorem}

The key consequence of Theorem~\ref{thm:monotonenew} is as follows: Assuming system~\eqref{eq:bivirus:controlnew} has a finite number of equilibria (and it can be argued that this assumption is fulfilled with generic parameter values), the typical behavior of system~\eqref{eq:bivirus:controlnew} is convergence to some stable equilibrium. Typical here means from almost all initial conditions. An initial condition coinciding with an unstable equilibrium is a nontypical condition, as is an initial condition in which one of $x^1(0)$ and $x^2(0)$ is zero. Limit cycles, if they exist, must be nonattractive. All other behaviors such as chaos can be conclusively ruled out \cite{sontag2007monotone}.

The presence of control does however produce changes in the way equilibria occur. 
The paper will now focus more rigorously on such changes, beginning with the more straightforward single population case.

\section{The single population case} 
\label{sec:n=1case}
In this section, we investigate the behavior of a single-population bivirus SIS model in the presence of the controller of interest. Our interest in considering the single-population case stems from the observation that its simplicity makes analysis tractable, and even so, the results obtained 
serve as a baseline and additional motivation for extensions to the networked case. As we will see later in this section, a crucial difference between the uncontrolled  system 
 and the controlled system 
  is that the former does not admit any coexistence equilibria, whereas the latter does. A downstream effect of said difference is with regard to a 
  ``winner takes all"
  battle. 
  In the language of dynamical systems, the winner of the aforementioned battle is the virus whose single-virus endemic equilibrium is stable. For the uncontrolled system with generic parameter choices, since there cannot be any coexistence equilibrium (let alone a stable coexistence equilibrium), when two viruses persist in the population, one of them must necessarily win \cite{zino2024modeling}. However, for the controlled system, the possibility of the presence of a coexistence equilibrium implies a situation where there is \emph{no} winner in a  winner takes all 
  battle.

\subsection{An illustrative example} Before proceeding with a formal analysis of the single-population setting, we present an illustrative example by way of simulation. 
Suppose that $\delta^1 = \delta^2 = 1$, $b^1 = 3$ and $b^2 = 2$. In the uncontrolled dynamics of \eqref{eq:uncontrolled_single}, 
we observe convergence from all initial conditions in the interior of $\mathcal D$ to the virus~1 boundary equilibrium (i.e., $(\bar{x}^1, 0)$), as per Fig.~\ref{fig:uncontrolled}. We then introduce control $h^1(x^1) = \alpha x^1$, while setting $h^2(x^2) \equiv 0$. 
In Fig.~\ref{fig:alpha_05} -- \ref{fig:alpha_3}, we progressively increase $\alpha$, and observe that the system instead converges to a coexistence equilibrium. It appears that there is some $\bar \alpha$ value (here 1), below which 
the controlled system continues to converge to the virus~1 boundary equilibrium, and above which it converges to a coexistence equilibrium. As noted in the Introduction, in the absence of control, the system can never admit any coexistence equilibrium (except in the nongeneric situation 
i.e., $b^1/\delta^1=b^2/\delta^2$), showing that feedback control has fundamentally changed the system and its equilibria. 

\begin{figure}
    \centering
        \subfloat[Uncontrolled dynamics]{\includegraphics[width=0.5\columnwidth]{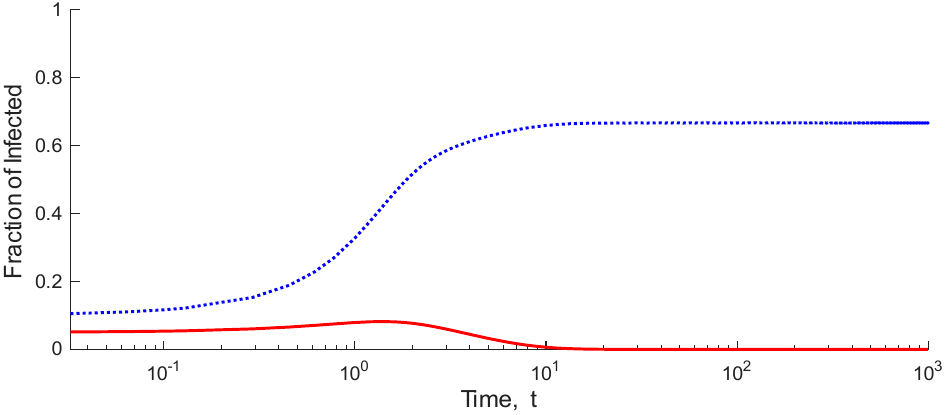}\label{fig:uncontrolled}}
        \subfloat[ $\alpha = 0.5$.]{\includegraphics[width=0.5\columnwidth]{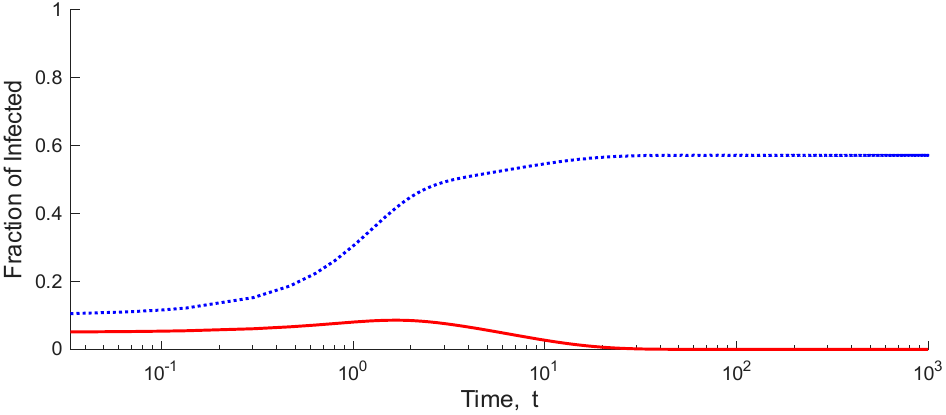}\label{fig:alpha_05}}%
        \vfill
        \subfloat[ $\alpha = 1.05$.]{\includegraphics[width=0.5\columnwidth]{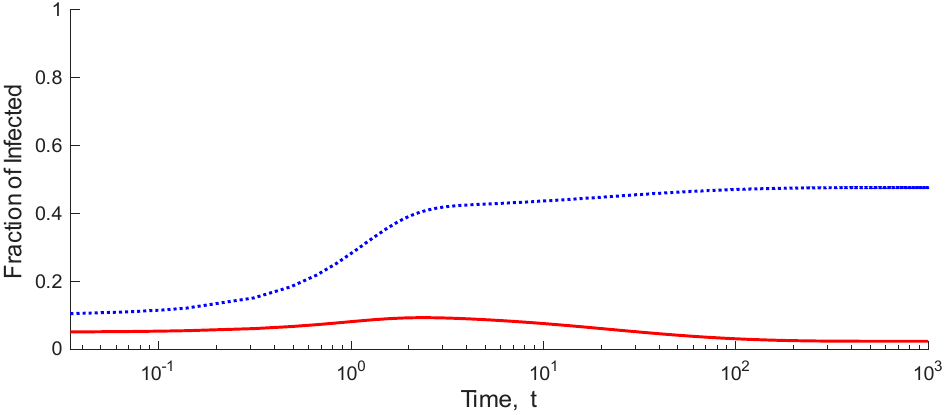}\label{fig:alpha_105}}
        \subfloat[ $\alpha = 3$.]{\includegraphics[width=0.5\columnwidth]{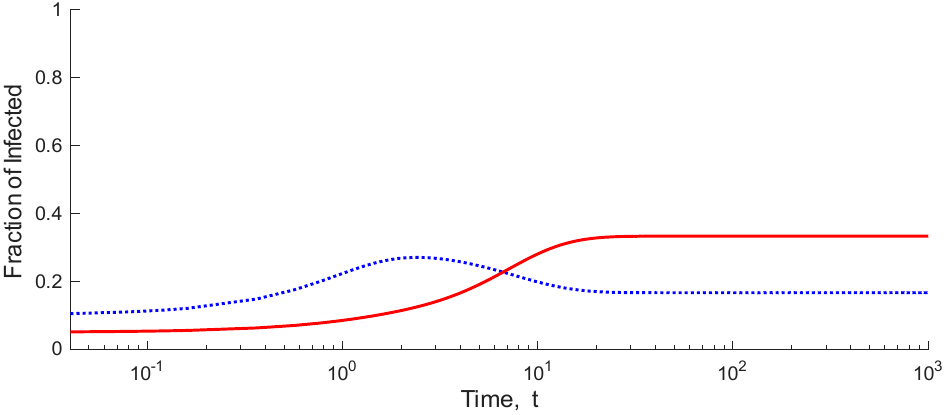}\label{fig:alpha_3}}%
    \caption{Simulations for the $n = 1$ case\label{fig:n1}. Virus 1 (blue line), Virus 2 (red line). }
\end{figure}

\subsection{Formal Analysis}\label{sec:n=1:formal:analysis}

Our starting point is the pair of equations \eqref{eq:controlledsingle}. Suppose that virus 1 is stronger, i.e. $b^1/\delta^1>b^2/\delta^2>1$. With no control, a `winner takes all' attractive equilibrium occurs at $(1-\delta_1/b_1,0)$ on the virus 1 axis. A control is introduced to attempt to lower the fraction of people infected with virus~1. No control is applied to virus 2. More specifically, we consider the controlled set
\begin{align}
\dot x^1(t)&=[-\delta^1+(1-x^1(t)-x^2(t))b^1-\alpha x^1(t))]x^1(t)\nonumber\\
\dot x^2(t)&=[-\delta^2+(1-x^1(t)-x^2(t))b^2]x^2(t)  \label{eq:controlledsinglealpha}
\end{align}
where $\alpha$ is nonnegative (the zero value corresponding to no control). 

It is straightforward to determine all equilibria of these equations, as well as their stability properties. First, the positions are:
\begin{align}
    E_0&=(0,0);\quad E_1=\big(\frac{b^1-\delta^1}{b^1+\alpha},0\big);\quad E_2=\big(0,1-\frac{\delta^2}{b^2}\big);\nonumber\\
    E_3&=\big (\frac{b^1}{\alpha}(\frac{\delta^2}{b^2}-\frac{\delta^1}{b^1}),1-\frac{\delta^2}{b^2}-\frac{b^1}{\alpha}(\frac{\delta^2}{b^2}-\frac{\delta^1}{b^1})\big)
\end{align}

The 
DFE $E_0$ is to be expected. The virus 1-only equilibrium $E_1$ becomes smaller, the bigger is $\alpha$, which is consistent with the anticipated goal of using control and with previous findings~\cite{ye2021_PH_TAC}. The virus 2-only equilibrium $E_2$ is unaltered in its position. The remaining equilibrium $E_3$ is an additional solution of the steady state equation associated with \eqref{eq:controlledsinglealpha} arising because $\alpha\neq 0$ but absent when $\alpha=0$. To understand why there can be no coexistence equilibrium when $\alpha=0$, observe that the steady state equations with the assumption that $\bar x^1\neq0,\bar x^2\neq 0$ yield two linear equations for $x^1+x^2$ which are incompatible, because $\delta^1/b^1\neq\delta^2/b^2$. For small values of $\alpha$, $x_1$ will be very large and $x_2$ has a negative value. Thus in relation to modeling of bivirus spread, $E_3$ is outside the region $\mathcal D$ of interest. However, as $\alpha$ increases, there is a critical value of $\alpha$ where $x_1$ becomes less than 1 and a larger value again where the sign of $x_2$ changes to positive. Above this latter value, $E_3$ becomes a coexistence equilibrium inside $\mathcal D$.  This critical value is

\begin{equation}\label{eq:criticalalpha}
\bar \alpha=\frac{b^1\delta^2-b^2\delta^1}{b^2-\delta^2}.
\end{equation}
While for $\alpha\geq\bar\alpha$, there is a coexistence equilibrium in $\mathcal D$,  at $\alpha= \bar\alpha$ there holds 
\begin{equation}
    E_1=E_3=\big(1-\frac{\delta^2}{b^2},0\big)
\end{equation}

To sum up, as $\alpha$ varies upwards from zero, two of the equilibria $E_0$ and $E_2$ remain fixed, while the other two $E_1$ and $E_3$ move on straight lines. One of these lines is part of the $x_1$ axis, and the other crosses it. At the intersection point, two of the equilibria temporarily coincide. 

 We now  consider the associated stability properties of the four equilibria, first considering the eigenvalues at each equilibrium of the Jacobian associated with \eqref{eq:controlledsinglealpha}. 

 The Jacobian associated with \eqref{eq:controlledsinglealpha} at a general point $(x^1,x^2))$ is given by
\begin{small}
 \begin{align}\label{eq:generalJacobiansinglealpha}
&J(x^1,x^2)=\nonumber\\
&\begin{bmatrix}
    b^1-\delta^1-2(b^1+\alpha)x^1-b^1x^2&-b^1x^1\\
    -b^2x^2&b^2-\delta^2-b^2x^1-2b^2x^2
\end{bmatrix}
\end{align}
\end{small}

At $E_0$ and $E_2$, the expressions are independent of $\alpha$ and correspond to an unstable node and a saddle node:
\begin{align}
    J(E_0)&=\begin{bmatrix}
            b^1-\delta^1&0\\0&b^2-\delta^2
            \end{bmatrix}
        \\\nonumber
    J(E_2)&=\begin{bmatrix}
            \frac{b^1\delta^2}{b^2}-\delta^1&0\\
            -b^2+\delta^2&-b^2+\delta^2
            \end{bmatrix}
\end{align}
    
Next, we have
\begin{equation}
    J(E_1)=\begin{bmatrix}
            -b^1+\delta^1&-b^1\frac{b^1-\delta^1}{b^1+\alpha}\\
            0&b^2-\delta^2-b^2\frac{b^1-\delta^1}{b^1+\alpha}
    \end{bmatrix}
\end{equation}

One eigenvalue, viz. $-b^1+\delta^1$ is always negative. The other is negative for $\alpha<\bar \alpha$, zero for $\alpha=\bar\alpha$ and positive for $\alpha>\bar\alpha$. Thus increasing $\alpha$ changes its character from stable to unstable. 

Last, at the coexistence equilibrium $E_3=(\bar x^1_3,\bar x^2_3)$ (but irrespective of the sign of $\bar x^2_3$). the equilibrium equations allow simplification of the expression for the Jacobian and
\begin{equation}
    J(E_3)=\begin{bmatrix}
        -(b^1+\alpha)\bar x^1_3&-b^1\bar x^1_3\\-b^2\bar x^2_3&-b^2\bar x^2_3
    \end{bmatrix}
\end{equation}
Inside $\mathcal D$, the values of $\bar x^1_3$ and $\bar x^2_3$ are constrained and force $J$ to have a positive determinant, and negative trace, implying stability. When $\bar x^2_3$ is negative, instability occurs due to the negative determinant, but this is outside of $\mathcal D$ and thus of limited relevance.

In summary, as $\alpha$ increases from zero, the virus-1 boundary equilibrium point $E_1$ changes from being stable to unstable, and a stable coexistence equilibrium appears in $\mathcal D$. This is exactly the form a transcritical bifurcation takes. In this case, all stable coexistence equilibria have the same value of $\bar x^1+\bar x^2$, i.e. the total fraction of the population that is infected with either virus is independent of the value of $\alpha$; increase of $\alpha$ beyond $\bar\alpha$ solely serves to change the mix of infections between those due to the two different viruses. See Fig.~\ref{fig:bifurcation_diagram} for a visual representation.

All these observations are of course consistent with the example given earlier, in Fig.~\ref{fig:n1}.

We can also observe that once $\alpha$ is such that the two boundary equilibria are both unstable, it is a consequence of monotone systems theory that there is necessarily a stable coexistence equilibrium \cite[Corollary~6.3]{doshi2022convergence}. The conclusion is also consistent with the way Poincar\'e-Hopf theory can be used to pin down aspects of the equilibria and their stability properties, \cite{anderson2023equilibria}.


\begin{figure}
    \centering
    \resizebox{\columnwidth}{!}{\begin{tikzpicture}[>=Latex, line cap=round, line join=round]

    \definecolor{ptgreen}{RGB}{0,153,0}
    \definecolor{ptblue}{RGB}{0,0,238}
    \definecolor{ptred}{RGB}{237,28,36}
    \definecolor{ptmagenta}{RGB}{204,51,204}

    \draw[->, thick] (0,0) -- (6.8,0) node[right] {\footnotesize$x^1$};
    \draw[->, thick] (0,0) -- (0,3.3) node[above] {\footnotesize$x^2$};

    \coordinate (O)      at (0,0);
    \coordinate (Ptop)   at (0,1.6);    
    \coordinate (Ptri)   at (3.2,0);    
    \coordinate (Pright) at (5.0,0);    
    \coordinate (Bend)   at (0,3.0);    
    \coordinate (Aend)   at (6.0,0);    
    \coordinate (Rend)   at (5.0,-0.9); 

    \draw[thick] (Bend) -- (Aend);
    \node[left] at (Bend) {\footnotesize$(0,1)$};

    \draw[ptred, very thick] (O) -- (Ptri);        
    \draw[ptred, very thick] (Ptri) -- (Rend);     

    \draw[ptblue, very thick] (Ptop) -- (Ptri);    
    \draw[ptblue, very thick] (Ptri) -- (Pright);  

    \node[fill=ptgreen, draw=ptgreen, minimum size=6pt, inner sep=0pt] at (Ptop) {};
    \node[fill=ptgreen, draw=ptgreen, minimum size=6pt, inner sep=0pt] at (Pright) {};
    \node[regular polygon, regular polygon sides=3, fill=ptmagenta,
          draw=ptmagenta, minimum size=9pt, inner sep=0pt] at (Ptri) {};

    \node[right] at (0,1.8) {\footnotesize$1-\dfrac{\delta^2}{b^2}$};
    \node[below=2pt] at (5.5,0.15) {\footnotesize$1-\dfrac{\delta^1}{b^1}$};
    \node at (6.3,0.25) {\footnotesize$(1,0)$};

    \draw[->] (1.868,0.8) -- (1.332,1.07);
    \node at (2.15,0.98) {\footnotesize$\alpha>\bar\alpha$};

    \draw[->] (1.5,0.12) -- (1,0.12);
    \node[above] at (1.1,0.12) {\footnotesize$\alpha>\bar\alpha$};

    \draw[->] (4.6,0.12) -- (4.1,0.12);
    \node[above] at (4.35,0.12) {\footnotesize$\alpha<\bar\alpha$};

    \draw[->] (4.368,-0.7) -- (3.832,-0.43);
    \node[below] at (4.1,-0.72) {\footnotesize$\alpha<\bar\alpha$};

    \node[below left] at (2.9,-0.1) {\footnotesize$1-\dfrac{\delta^2}{b^2}$};

  \end{tikzpicture}%
}
    \caption{Transcritical bifurcation. The green squares indicate the single-virus boundary equilibria; the upper left square is stationary for all $\alpha$, while the lower right square captures $\alpha =  0$ and it slides left as $\alpha$ increases. Direction of black arrows indicate \textit{increasing} $\alpha$ for 
    system \eqref{eq:controlledsinglealpha}. Blue and red lines correspond to equilibrium points for differing $\alpha$; blue and red indicate a stable and unstable equilibrium, respectively.  The transcritical bifurcation occurs at the pink triangle, with $\alpha = \bar{\alpha}$, at the $x^1$ value of $1- \frac{\delta^2}{b^2}$. }
    \label{fig:bifurcation_diagram}
\end{figure}
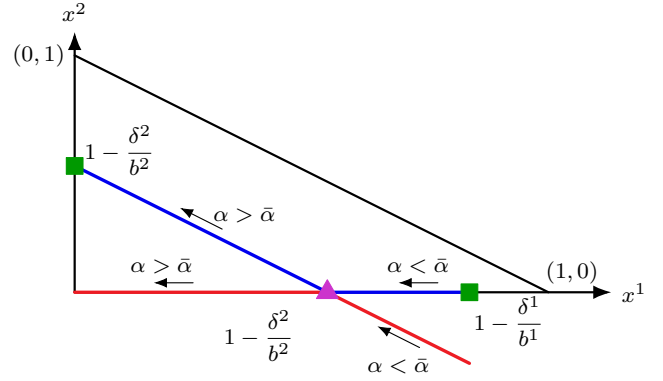

 \section{Equilibria Stability Analysis for the general case - DFE and Boundary Equilibria}\label{sec:equi:analysis}
 In this section, we consider the multi-population case (i.e., system~\eqref{eq:bivirus:controlnew}), and identify conditions for stability of the DFE and show that an unstable boundary equilibrium of system~\eqref{eq:bivirus} cannot be stabilized using the family of decentralized  state feedback controllers given in~\eqref{eq:dist.contrller}.
 
 First, observe that the DFE ($\textbf{0}, \textbf{0}$) is always an equilibrium for system~\eqref{eq:bivirus:controlnew}. For system~\eqref{eq:bivirus:controlnew}, specialized for the single-virus case and assuming $s(-D+B)\leq0$, the DFE (i.e, the point $\textbf{0}$) is a globally asymptotically stable equilibrium, see \cite[Theorem~3]{ye2021_PH_TAC}. It is quite straightforward to show that the condition in \cite[Theorem~3]{ye2021_PH_TAC} can be generalized to secure global asymptotic stability of the DFE. We have the following result.
 \begin{proposition}\label{prop:dfe_bivirus_control}
    Consider system~\eqref{eq:bivirus:controlnew} under Assumptions~\ref{assum:base} and~\ref{assum:irreducible}.  Suppose that for all $i \in [n]$ and $k \in [2]$, $h_i^k
: [0, 1] \rightarrow \mathbb{R}_{\geq 0}$ is
bounded, smooth and monotonically nondecreasing, satisfying $h_i^k(0) = 0$. The following statements hold:
\begin{enumerate}
    \item If $s(-D^k+B^k)\leq 0$, then for all $(x^1(0), x^2(0))\in \mathcal D$, one has $\lim_{t \rightarrow \infty}x^k(t)=\textbf{0}$.
    \item The DFE is the unique equilibrium if, and only if, $s(-D^k+B^k)\leq 0$ for $k=1,2$.
\end{enumerate}
\end{proposition}
\begin{proof}
    We begin with statement~1. From Lemma~\ref{lem:pos:inv:D:control}, we know that $x^1(t)+x^2(t) \in [0,1]$ for all $t$. This implies that $(1-x_i^1(t)-x_i^2(t)) \leq (1-x_i^k(t))$  for $k = 1,2$. Therefore, for each $i \in [n]$, we can bound each scalar equation in \eqref{eq:bivirus:controlnew} as
 \begin{align} 
     \dot{x}_i^1(t) &\leq -\delta_i^1x_i^1(t)+(1-x_i^1(t))\sum_{j=1}^nb_{ij}^1x_j^1(t), \\
      \dot{x}_i^2(t) &\leq -\delta_i^2x_i^2(t)+(1-x_i^2(t))\sum_{j=1}^nb_{ij}^2x_j^2(t),
 \end{align}
 because $h_i^k(x_i^k(t))x_i^k(t) \geq 0$ and $(1-x_i^k(t)) \geq 0$ for all $t$. In fact, this means that the left side of \eqref{eq:bivirus:controlnew} is bounded as
\begin{subequations}\label{eq:bound}
\begin{align}
   \dot{x}^1(t) & \leq  \big(- D^1+(I-X^1(t))B^1 \big)x^1(t) \\
      \dot{x}^2(t) & \leq  \big(- D^2+(I-X^2(t))B^2 \big)x^2(t). 
\end{align} 
\end{subequations} 

Now, consider the single virus system
\begin{equation}\label{eq:sis_system}
    \dot y(t) = \big(-D^k+(I-Y(t))B^k\big)y(t).
\end{equation}
It is well known that \eqref{eq:sis_system} converges to $y = {\bf 0}$ if and only if $s(-D^k+B^k) \leq 0$, with the convergence being exponentially fast for a strict inequality, and at a rate of $t^{-1}$ for equality. See e.g.~\cite{ye2024rate}. Now, define the function
\begin{equation}
    g^k(x) = (-D^k + (I-X)B^k)x + M^kx,
\end{equation}
where $X = \diag(x)$, and we pick the 
scalar $M^k$ such that $M^k >  \max_{i} [\delta_i^k +\sum_{j=1}^n b_{ij}^k]$. One can verify that $x\leq y$ implies $g(x) \leq g(y)$, by noting that the Jacobian of $g(x)$ is a nonnegative matrix and then applying the mean-value theorem. This means that \eqref{eq:sis_system} satisfies Condition~Q in \cite[Section~2, pg. 256]{walter1971ordinary}, and thus the main theorem of \cite{walter1971ordinary} (see Section~3 in \cite{walter1971ordinary}) can be applied. 
In particular, defining $f^k(y) = (-D^k+(I-Y)B^k)y$, where $Y = \diag(y)$, observe that our above inequality bounds combine and culminate in the expression
\begin{equation}
    \dot x^k(t) - f(x^k(t)) \leq \dot y(t) - f(y(t)),
\end{equation}
and thus $x(t) \leq y(t)$ in $\mathcal D^k$. Thus, $x(t)\to{\bf 0}$ as $t\to\infty$.

For statement~2, sufficiency is immediately obtained from applying statement~1 for both $k=1$ and $k=2$, which ensures that $\lim_{t\to\infty} x^1(t) = x^2(t) = {\bf 0}$. To establish necessity, suppose to the contrary that $s(-D^1+B^1) > 0$. Now, consider \eqref{eq:bivirus:controlnew} with initial conditions $(x^1(0), {\bf 0})$, where $x^1(0) \in \mathcal D^k \setminus \{{\bf 0}\}$. The dynamics are now identical to the controlled single virus SIS model studied in \cite{ye2021_PH_TAC}, where it was established that $\lim_{t\to\infty} x^1(t) = \bar x^1$, where $\bar x^1 \gg {\bf 0}$ is the controlled endemic equilibrium. In other words, \eqref{eq:bivirus:controlnew} has at least one other equilibrium besides the DFE, being $(\bar x^1, {\bf 0})$.
\end{proof}
Based on Proposition~\ref{prop:dfe_bivirus_control}, it is clear that what remains of interest is the case where $s(-D^k+B^k)>0$ for $k=1,2$. Accordingly,  henceforth we invoke the following assumption. 

 \begin{assumption}\label{assum:spec:radii:more than:1}
 The parameters of system~\eqref{eq:bivirus:controlnew} are such that    $s(-D^k+B^k)>0$ for $k=1,2$.
 \end{assumption}

We have already observed that system~\eqref{eq:bivirus:controlnew} has exactly three equilibria on the boundary of the set $\mathcal D$ as in the control free case (see the text above Lemma~\ref{lem:equi_non-zero_nonone:1}). The next result observes that neither of the non-zero boundary equilibria is larger in magnitude than the corresponding non-zero boundary equilibrium of the uncontrolled system, i.e., system~\eqref{eq:bivirus}. Also, introducing control cannot convert an unstable boundary equilibrium to a stable one. 
The formal statement is as follows.
\begin{theorem}\label{thm:merged}
   Consider system~\eqref{eq:bivirus:controlnew} under Assumptions~\ref{assum:base}-\ref{assum:spec:radii:more than:1}.  Suppose that for all $i \in [n]$, the function $h_i^1
: [0, 1] \rightarrow \mathbb{R}_{\geq 0}$ is
bounded, smooth and monotonically nondecreasing, satisfying $h_i^1(0) = 0$ with at least one $h_i^1$ positive for positive values of its argument.  Let $\tilde{x}^1$ (resp. $\bar{x}^1$) 
denote the single-virus endemic equilibrium of virus~$1$ for system~\eqref{eq:bivirus:controlnew} (resp.~\eqref{eq:bivirus}).
    \begin{enumerate}[label=\roman*)]
        \item \label{q2} There holds $\tilde{x}^1 \ll \bar{x}^1$.
        \item\label{q3}  If $(\bar x^1,0)$ is an unstable equilibrium, then $(\tilde x^1,0)$ will be an unstable equilibrium. 
 
        \end{enumerate}
\end{theorem}
Note that the same conclusion will follow if the choice of virus 1 above is replaced by virus 2. 

\textit{Proof:}
(Part (i)) We will prove the claim for $k=1$. Dropping superscripts, we have
\begin{align}
[-D+(I-\bar X)B]\bar x&={\bf 0}, \\ [-D+(I-\tilde X)B-H(\tilde x)]\tilde x&={\bf 0}
\end{align}
We will use the facts that ${\bf{0}}\ll\bar x,\tilde x \ll {\bf{1}}$ (see Lemma \ref{lem:equi_non-zero_nonone:1}). Define
\begin{equation}
    r=\max_i\frac{\tilde x_i}{\bar x_i}\quad J=\{j:j=\arg\max\frac{\tilde x_i}{\bar x_i}\}
\end{equation}
Part (i) is proved by establishing that $r<1$. Now, these definitions, together with the entry-wise expressions of the equilibrium equations above, imply that 
\[\tilde x_j=r\bar x_j,j\in J\quad\tilde x_i<r\bar x_i, i\notin J\quad\mbox{and}\quad B\tilde x\leq rB\bar x
\]
Consider for any $j\in J$ now the inequality chain:
\begin{align*}
r\delta_j\bar x_j &\leq r(\delta_j+h_j(\tilde x_j))\bar x_j=(\delta_j+h_j(\tilde x_j))\tilde x_j\\\nonumber
&=(1-\tilde x_j)(B\tilde x)_j \leq (1-r\bar x_j)(rB\bar x)_j
\end{align*}
If $r>1$, we can push the last inequality further to obtain a contradiction:
\begin{align*}
    r\delta_j\bar x_j\leq (1-r\bar x_j)(rB\bar x)_j<r(1-\bar x_j)(B\bar x)_j=r\delta_j\bar x_j
\end{align*}
Hence we have established that $\tilde x\leq \bar x$. To obtain a further contradiction, assume now $r=1$; then tracing through the inequality chain again, now with the assumption that $\bar x_j=\tilde x_j$, yields $h_j(\tilde x_j)=0$ and $(B\tilde x)_j=(B\bar x)_j$. Since $H(\tilde x)$ cannot be zero, the set $J$ cannot coincide with $[n]$. This means that $y:=\bar x-\tilde x$ is a nonnegative but not positive vector with the property that if the $j$-th entry is zero, so is the $j$-th entry of $By$.   This contradicts the irreducibility of $B$. Hence $r<1$.

(Part (ii)) Using \eqref{eq:jacobian} and \eqref{eq:jacobianblocks}, we can obtain the Jacobian matrix associated with the boundary equilibrium $(\tilde x^1, {\bf 0})$. It is
\begin{equation*}
    P(\tilde x^1,{\bf{0}}))=\begin{bmatrix}
        P_{11}(\tilde x^1,0)&-\diag(B^1\tilde x^1)\\
        {\bf 0}&-D^2+(I-\tilde X^1)B^2
    \end{bmatrix}
    \end{equation*}
where
\begin{align}
    P_{11}(\tilde x^1,{\bf{0}})=&-D^1+(I-\tilde X^1)B^1-\diag(B^1\tilde x^1)\nonumber\\
            &-H^1(\tilde X^1)-H^1(\tilde x^1)'\tilde X^1\nonumber
\end{align}
Now the single virus problem for virus 1 has a single endemic equilibrium that is precisely $\tilde x^1$,  as noted in Lemma \ref{lem:yeLotka}, this equilibrium is stable. Since $P_{11}(\tilde x^1,{\bf{0}})$ is precisely the Jacobian associated with the single virus equation evaluated at this equilibrium, it follows that the eigenvalues of $P_{11}(\tilde x^1,{\bf{0}})$ all have negative real parts, independently of the precise form of $H$. Because of the block triangular structure of $P$, the stability of $(\tilde x^1,{\bf{0}})$ is therefore determined by the eigenvalues of $P_{22}(\tilde x^1)=-D^2+(I-\tilde X^1)B^2$. Likewise, the stability of $(\bar x^1,{\bf{0}})$ is determined by the eigenvalues of $P_{22}(\bar x^1)=-D^2+(I-\bar X^1)B^2$. It is evident that
\[
P_{22}(\tilde x^1)=P_{22}(\bar x^1)+(\bar X^1-\tilde X^1)B^2
\]
Now $P_{22}(\bar x^1)$ and $P_{22}(\tilde x^1)$ are Metzler matrices. 
Note that they differ by a nonnegative matrix, since $\tilde{x}^1 \ll \bar{x}^1$ (from statement~i)). Therefore 
it follows that their spectral abscissas obey $s(P_{22}(\tilde x^1))>s(P_{22}(\bar x^1))$. Instability of $(\bar x^1,0)$ as an equilibrium of the bivirus system implies that $s(P_{22}(\bar x^1)>0$, implying then $s(P_{22}(\tilde x^1)>0$ or instability of the controlled bivirus system equilibrium $(\tilde x^1,{\bf{0}})$.~\qed

The following is immediate.
\begin{corollary}
Statement~\ref{q3} of Theorem~\ref{thm:merged}  guarantees that an unstable boundary equilibrium of system~\eqref{eq:bivirus} remains so even when a distributed controller of the form given in~\eqref{eq:dist.contrller} acts on system~\eqref{eq:bivirus}.  
\end{corollary}

The theorem and this corollary are consistent with what was observed in the $n=1$ case with explicit analysis. (Note for the $n=1$ case, we did not look at the effect of including control on the virus whose endemic equilibrium was unstable). In particular, we saw that the use of control on the dominant virus could reduce the infected fraction at steady state, but beyond a certain point, the `winner takes all' idea fails, and while a stable coexistence equilibrium occurs, the single virus equilibrium of the dominant virus becomes unstable. 

Just as in the single population case, even for the $2-$ population case, we can find a situation where a boundary equilibrium loses stability and a stable coexistence equilibrium comes into existence with a sufficiently strong control. 
    Consider the following example of a 2-population bivirus system, where $D^1=D^2=I$, $B^1=[\begin{smallmatrix}
        1.6&1\\1&1.6
    \end{smallmatrix}]$, and $B^2=[\begin{smallmatrix}
        2.1&0.885\\
        1.885&1.1
    \end{smallmatrix}]$. Exact details are reported in \cite[Case~4, S338]{ye2021convergence}, but in summary, the single-virus endemic equilibrium corresponding to virus~1 is globally stable, and there is no coexistence equilibrium; see Figure~\ref{fig:n2_uncontrolled_case4}. We 
    set $H^1(x^1)=\alpha x^1,H^2(x^2)=0$ for gain $\alpha>0$. As we increase $\alpha$, the single-virus endemic equilibrium of virus~1 becomes unstable, and a coexistence equilibrium emerges and becomes  attractive; see Figures \ref{fig:n2_alpha05_case4}  and \ref{fig:n2_alpha5_case4}. Additional simulations reveal that for several initial conditions, the dynamics do converge to this coexistence equilibrium, thus suggesting that it is globally attractive.  However, our simulations also indicate that no matter how much we increase $\alpha$, the dynamics do not converge to the single-virus endemic equilibrium of virus~2.
\begin{figure*}
    \centering
        \subfloat[Uncontrolled dynamics]{\includegraphics[width=0.3\linewidth]{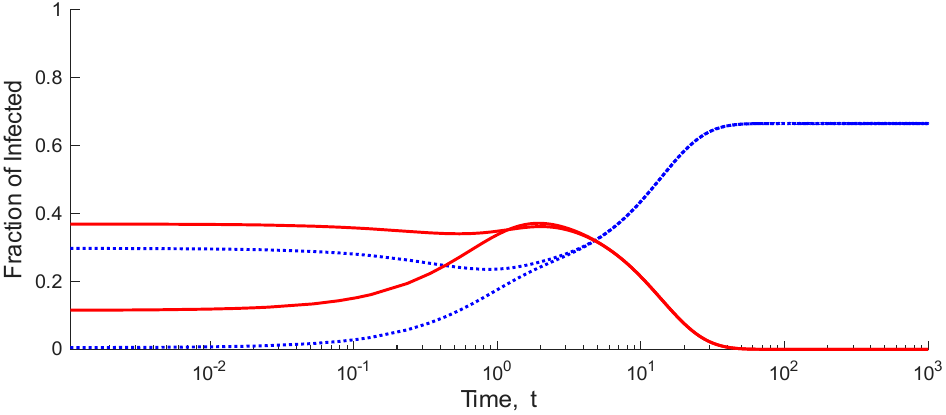}\label{fig:n2_uncontrolled_case4}}
        \subfloat[$\alpha = 0.5$]{\includegraphics[width=0.3\linewidth]{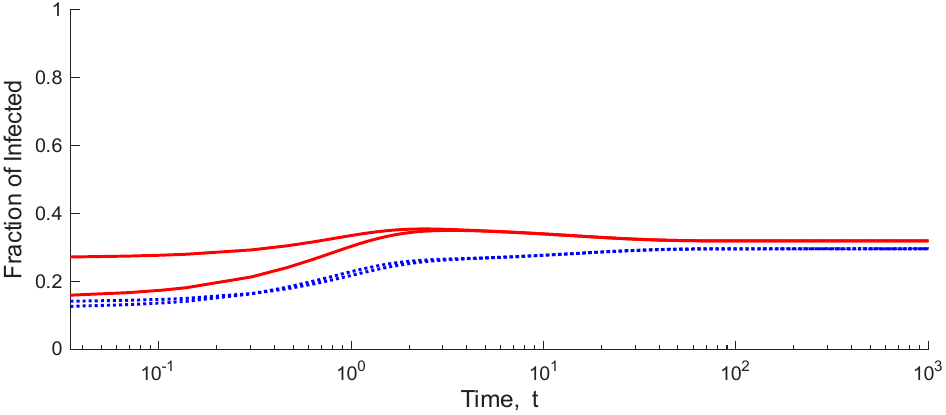}\label{fig:n2_alpha05_case4}}%
        \subfloat[$\alpha = 5$]{\includegraphics[width=0.3\linewidth]{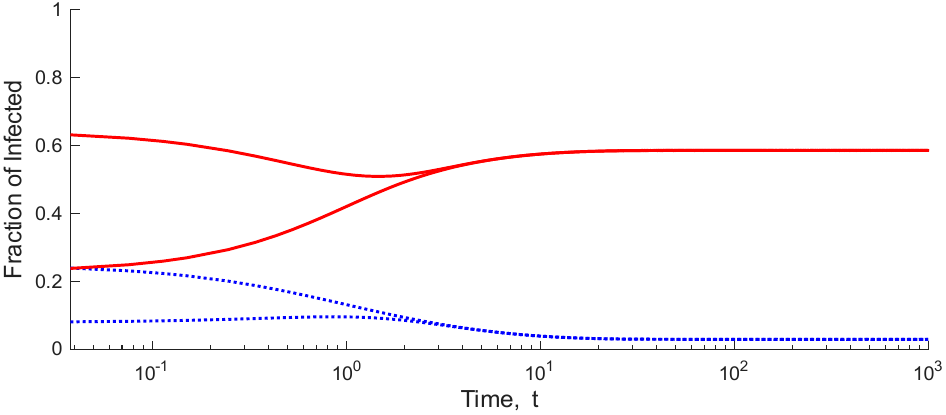}\label{fig:n2_alpha5_case4}}
    \caption{Simulations for $n = 2$ for \cite[Case 4, S338]{ye2021convergence}. Virus 1 (blue line), Virus 2 (red line).\label{fig:n2_case4} }
\end{figure*}

As a second simulation with two populations, we use 
\cite[Case~2, S337]{ye2021convergence}, except with $B^1$ and $B^2$ switched. In the uncontrolled dynamics, there are actually two locally stable boundary equilibria (namely, $(\bar{x}^1, \textbf{0})$ and $(\textbf{0}, \bar{x}^2)$) and one unstable coexistence equilibrium, as per Fig.~\ref{fig:n2_uncontrolled}. Thus, in contrast to the single population case, one does not necessarily have a `winner takes all' boundary equilibrium in the control-free multipopulation case.  However, when we set $\alpha = 0.1$ (and any larger $\alpha$), we instead always converge to the virus~2 boundary equilibrium (i.e., $(\textbf{0}, \tilde{x}^2)$) (unless we are already at another equilibrium), and one can check that the virus~1 boundary equilibrium (i.e., $( \tilde{x}^1, \textbf{0})$) becomes unstable; see Figure~\ref{fig:n2_alpha01}.

\begin{figure}
    \centering
        \subfloat[Uncontrolled dynamics]{\includegraphics[width=0.5\columnwidth]{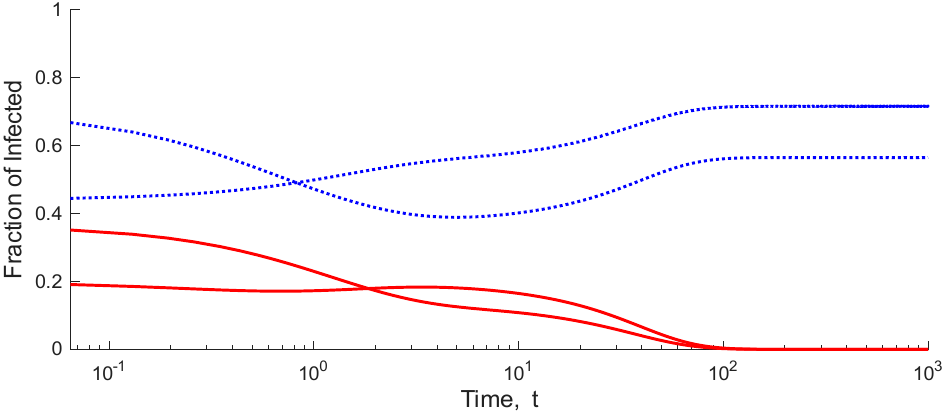}\label{fig:n2_uncontrolled}}%
        \subfloat[$\alpha = 0.1$]{\includegraphics[width=0.5\columnwidth]{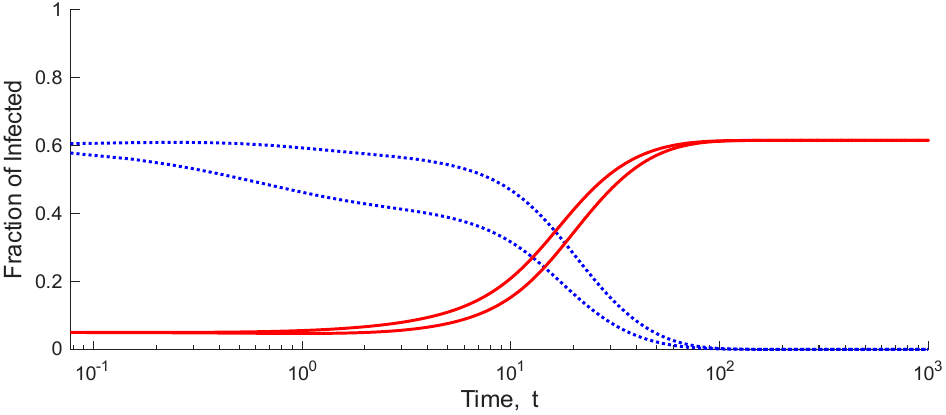}\label{fig:n2_alpha01}}
    \caption{Simulations for $n = 2$ for \cite[Case~2, S337]{ye2021convergence}. Virus~1 (blue), virus~2 (red). 
    \label{fig:n2_case2} }
\end{figure}

 \section{Conclusion}\label{sec:conclusion}
The  paper focused on the classic bivirus networked SIS system that is acted upon by a class of  decentralized state feedback controllers. 
We showed that 
the controlled bivirus SIS system 
exhibits novel features vis-a-vis the uncontrolled bivirus SIS system. In particular, for the single population case, the uncontrolled bivirus system does not admit isolated coexistence equilibria, while the controlled bivirus system can. Further, for general networked systems, the given controller cannot stabilize an unstable boundary equilibrium of the uncontrolled bivirus system.  Indeed, application of control can actually introduce a a form of instability, replacing a stable `winner takes all' equilibrium by a stable coexistence equilibrium through a transcritical bifurcation. 
Our ongoing work seeks to provide a deeper understanding of the differences between the uncontrolled system and the controlled system, specifically, by providing analytical expressions for equilibria and Jacobians when using simultaneous control of both viruses, and possibly extend the proposed controller for time-varying SIS epidemics.
 
\bibliography{References}
\end{document}